\documentclass[sigconf]{acmart}

\makeatletter                   
\def\mdseries@tt{m}             
\makeatother                    
\usepackage[draft=true]{minted} 
\usepackage{color}
\usepackage{hyperref}           
\hypersetup{
    colorlinks=true,
    linkcolor=blue,
    filecolor=red,      
    urlcolor=magenta,
    breaklinks=true,            
}
\usepackage{breakurl}           

\AtBeginDocument{%
  \providecommand\BibTeX{{%
    \normalfont B\kern-0.5em{\scshape i\kern-0.25em b}\kern-0.8em\TeX}}}

\setcopyright{acmcopyright}
\copyrightyear{2020}
\acmYear{2020}
\acmDOI{10.1145/1122445.1122456}

\acmConference[ISSAC '20]{ISSAC '20: International Symposium on Symbolic and Algebraic Computation}{June 20--23, 2020}{Kalamata, Greece}
\acmBooktitle{ISSAC '20: International Symposium on Symbolic and Algebraic Computation, June 20--23, 2020, Kalamata, Greece}
\acmPrice{15.00}
\acmISBN{978-1-4503-9999-9/18/06}

\usepackage{multirow}
\usepackage{amsmath}
\usepackage{amscd}
\usepackage{amssymb}
\usepackage{amsfonts}
\usepackage{color}
\usepackage{algorithm}
\usepackage{algpseudocode}

\usepackage{blkarray}
\usepackage{graphicx}
\usepackage{bm}
\usepackage{xparse}
\usepackage{relsize}
\usepackage{tikz}
\usetikzlibrary{math}
\usetikzlibrary{patterns}
\usetikzlibrary{matrix,backgrounds}
\usetikzlibrary{positioning}
\usetikzlibrary{shapes.geometric,arrows}
\tikzset{fontscale/.style = {font=\relsize{#1}}}
\tikzset{mycolor/.style = {line width=1bp,color=#1}}%
\tikzset{myfillcolor/.style = {draw,fill=#1}}%
\pgfdeclarelayer{myback}
\pgfsetlayers{myback,background,main}

\newtheorem{theorem}{Theorem}[section]

\newtheorem{cor}[theorem]{Corollary}
\newtheorem{prop}[theorem]{Proposition}
\newtheorem{lemma}[theorem]{Lemma}
\newtheorem{defn}[theorem]{Definition}

\newtheorem{ex}[theorem]{Example}

\NewDocumentCommand{\highlight}{O{blue!40} m m}{%
\draw[mycolor=#1] (#2.north west)rectangle (#3.south east);
}
\NewDocumentCommand{\fhighlight}{O{blue!40} m m}{%
\draw[myfillcolor=#1] (#2.north west)rectangle (#3.south east);
}

\newcommand{\cE}{{\mathcal{E}}}
\newcommand{\cF}{{\mathcal{F}}}
\newcommand{\N}{{\mathbb{N}}}
\newcommand{\Z}{{\mathbb{Z}}}
\newcommand{\Q}{{\mathbb{Q}}}

\newcommand{\bI}{{\mathbb{I}}}

\newcommand{\hp}{{\rm{hp}}}
\newcommand{\hc}{{\rm{hc}}}
\newcommand{\hmm}{{\rm{hm}}}
\newcommand{\si}{{\rm{si}}}
\newcommand{\scc}{{\rm{sc}}}
\newcommand{\sv}{{\rm{sv}}}

\newcommand{\spa}{{\rm{span}}}

\newcommand{\ind}{{\rm{ind}}}
\newcommand{\dg}{{\text{deg}}}
\newcommand{\Li}{{\rm{Li}}}

\begin{document}

\title{An Additive Decomposition in S-Primitive Towers}



\author{Hao Du$^1$, \quad Jing Guo$^2$, \quad Ziming Li$^2$, \quad  Elaine Wong$^1$}
\affiliation{%
 \institution{$^1$Johann Radon Institute (RICAM), Austrian Academy of Sciences, Altenberger Stra\ss e 69, 4040, Linz, Austria}}
\affiliation{%
\institution{$^2$Key Laboratory  of Mathematics and Mechanization, AMSS, Chinese Academy of Sciences }}
\affiliation{%
\institution{School of Mathematical Sciences, University of Chinese Academy of Sciences, Beijing, 100190, China}}
\affiliation{%
\institution{hao.du@ricam.oeaw.ac.at, \,\,  JingG@amss.ac.cn, \,\, zmli@mmrc.iss.ac.cn, \,\, elaine.wong@ricam.oeaw.ac.at}}
%
%

\renewcommand{\shortauthors}{Du, Guo, Li, Wong}

\begin{abstract}
We consider the additive decomposition problem in primitive towers and present an algorithm to decompose a function in an S-primitive tower as a sum of a derivative in the tower and a remainder which is minimal in some sense. Special instances of S-primitive towers include differential fields generated by finitely many logarithmic functions and logarithmic integrals. A function in an S-primitive tower is integrable in the tower if and only if the remainder is equal to zero. The additive decomposition is achieved by viewing our towers not as a traditional chain of extension fields, but rather as a direct sum of certain subrings. Furthermore, we can determine whether or not a function in an S-primitive tower has an elementary integral without solving any differential equations. We also show that a kind of S-primitive towers, known as logarithmic towers, can be embedded into a particular extension where we can obtain a finer remainder.
\end{abstract}



\maketitle

\section{Introduction}
\label{SECT:introduction}

We consider the integrability problem in some class $\cF$ of functions in $x$, where $\cF$ is assumed to be closed under addition and  the usual derivation $'=\frac{d}{dx}$.
For $f\in\cF$, we ask if the indefinite integral of $f$ belongs to $\cF$. Let $\cF^\prime:=\{g' \mid g\in\cF\}$. The problem can therefore be stated as follows:
\begin{equation}\label{EQ:integrability}
\text{Given } f\in\cF,  \text{decide if } f\in\cF^\prime.
\end{equation}

\vspace{1cm}

We can see that a positive answer to \eqref{EQ:integrability} tells us that we can compute $g \in \cF$ such that $f = g^\prime.$ If \eqref{EQ:integrability} produces a negative answer, then we say $f$ is \emph{not integrable} in $\cF$.

In the latter case, we would still like to be able to say something about the given function. Is there any information to help us understand how far off we are from being successful? The answer lies in the additive decomposition problem:
$$\text{Compute } g, r \in \cF \text{ such that } f=g^\prime+r,$$
where
\begin{itemize}
\item[(i)] $r$ is minimal in some sense;
\item[(ii)] $f \in \cF^\prime$ if and only if $r = 0$.
\end{itemize}
We call such an $r$ a \emph{remainder} of $f$ in $\cF$ and write $$f\equiv r \mod \cF^\prime.$$ So, it is clear that an algorithm for solving the problem of additive decomposition also provides a solution to the integrability problem. 
Elements in $\cF^\prime$ have a special form, indicating that most functions have nonzero remainders. Remainders help us find \lq\lq closed form\rq\rq\ expressions for integrals of elements in $\cF$, in the sense that the integrals belong to some extensions over $\cF$. They also play an important role in reduction-based methods for creative telescoping.

The first additive decomposition for the class $\cF=\mathbb{C}(x)$ is due to Ostrogradsky~\cite{Ostrogradsky1845} and Hermite~\cite{Hermite1872}. Given a rational function $f \in \cF$, they were able to compute a remainder $r \in \cF$ of $f$ such that $r$ is proper and has a squarefree denominator, and $r$ is minimal in the sense that if $f \equiv \tilde{r} \mod \cF'$ for some $\tilde{r} \in \cF$, then the denominator of $r$ divides that of $\tilde{r}$.

There has been rapid development of additive decompositions in both symbolic integration and summation in recent years~\cite{Abra1995, BCCLX2013, BCLS2018, CHKL2015, CKK2016, CHKK2018, DHL2018, Raab2012, Hoeven2020}. Most of the articles were motivated by computing telescopers based on reduction~\cite{BCCL2013}. In the cited literature, some classes of functions that were studied include hyperexponential~\cite{BCCLX2013}, algebraic~\cite{CKK2016}, Fuchsian D-finite~\cite{CHKK2018}, and D-finite~\cite{Hoeven2020}. Additive decomposition problems in these classes have been fully solved. We observe that the space of D-finite functions is not closed under composition or taking reciprocals. For example, $\log x$ is D-finite, but $\log(\log(x))$ and $1/\log(x)$ are not. In this paper, we consider a class of functions that is closed under these two operations.


Singer et al. in 1985 and then Raab in 2012 gave some decision procedures for finding elementary integrals in some Liouvillian extensions \cite{Raab2012, SSC1985} and in the extensions which contain some nonlinear generators \cite{Raab2012}. They recursively solve Risch differential equations until one of them has no solution, or else the integral can be found. In the implementation of Raab's algorithm, the former case outputs an integrable part and collects all nonzero terms that prevent the differential equations from having a solution.
Recently, Chen, Du and Li \cite{CDL2018} were able to construct  remainders in some primitive extensions (they termed them ``straight towers'' and ``flat towers'') without solving any differential equations.

In this article, we expand their work \cite{CDL2018} to ``S-primitive towers'', which can be neither straight nor flat.
 Instances for S-primitive towers include differential field extensions generated by finitely many logarithmic functions and logarithmic integrals. Moreover, we show that a logarithmic tower can be embedded in a \\well-generated logarithmic tower with the aid of logarithmic product and quotient rules. We can compute ``finer'' remainders in such an extension.

\begin{figure}[ht]
\centering
\resizebox{8cm}{!}{
\begin{tikzpicture}[thick,scale=1, every node/.style={scale=1.1}]
\tikzmath{\xposleft=-6; \xposright=6; \xpos=0; \ypos=0; \sizefont=4;}

\draw (\xposleft-0.3,\ypos) ellipse (5.3 and 3.5);
\node at (\xposleft,\ypos+2.4)[fontscale=\sizefont] {Primitive Towers};
\node at (\xposleft,\ypos+1.7)[fontscale=\sizefont] {$K_0(t_1,\ldots,t_n)$};
\draw[fill=gray!60] (\xposleft,\ypos-1) ellipse (4.2 and 2);
\node at (\xposleft,\ypos+0.1)[fontscale=\sizefont] {S-Primitive Towers};
\draw[] (\xposleft+0.9,\ypos-1.4) ellipse (2.4 and 1.1);
\node at (\xposleft+0.1,\ypos-0.9) [fontscale=\sizefont]{Log};
\draw[] (\xposleft-2.4,\ypos-1.6) ellipse (2.2 and 0.7);
\node at (\xposleft-2.7,\ypos-1.6)[fontscale=\sizefont] {Straight};
\draw[] (\xposleft+1.6,\ypos-1.6) ellipse (1.4 and 0.5);
\node at (\xposleft+1.8,\ypos-1.6)[fontscale=\sizefont] {Flat};

\draw[fill=gray!60] (\xposright,\ypos+0.5) ellipse (3.3 and 1.5);
\node at (\xposright,\ypos+1)[fontscale=\sizefont] {Well-Generated};
\node at (\xposright,\ypos+0.4)[fontscale=\sizefont] {Log Towers};
\node at (\xposright,\ypos-0.2)[fontscale=\sizefont] {$K_0(u_1,\ldots,u_w)$};

\draw[->, ultra thick] (\xpos-5.2,\ypos-0.9) to [out=20,in=170] (\xpos+4,\ypos+0.5);
\node at (\xpos+0.7,\ypos+2) [fontscale=\sizefont] {Embedding};
\node at (\xpos+0.7, \ypos+1.4) [fontscale=\sizefont] {Theorem \ref{TH:wg}};
\end{tikzpicture}
}

\caption{The gray ellipses on the left indicate the classes of functions for which we can construct a remainder. The embedding gives us a field extension ($n\leq w$) where a ``finer'' remainder can be obtained.}
\label{FIG:mainresult}

\Description[short]{long}
\end{figure}
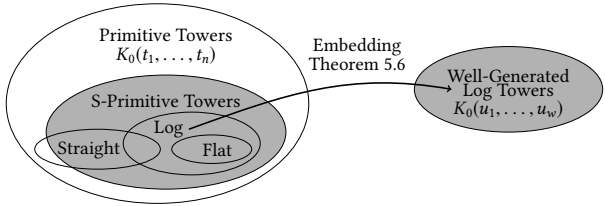

The organization of this article is as follows. In Sections~\ref{SECT:preliminaries} and \ref{SECT:md}, we give some relevant definitions associated to primitive towers, and then present a different way to view the towers. In Section~\ref{SECT:ad}, we give an algorithm for additive decompositions in S-primitive towers, and present a criterion for elementary integrability for the functions in such a field. In Section~\ref{SECT:wglt}, we discuss how to find a finer additive decomposition in well-generated logarithmic towers. Concluding remarks are given in Section~\ref{SECT:conc}.

\section{Preliminaries}
\label{SECT:preliminaries}

Let $K$ be a field of characteristic zero and $K(t)$ be the field of rational functions in $t$ over $K$. An element of $K(t)$ is said to be {\em $t$-proper} if the degree of its denominator in $t$ is higher than that of its numerator. In particular, zero is $t$-proper. For each $f \in K(t)$, there exists a unique $t$-proper element $g \in K(t)$ and a unique polynomial $p \in K[t]$ such that \begin{equation} \label{EQ:decomposition} f = g + p. \end{equation}

Let $^\prime$ be a derivation on $K$. The pair $(K, \, ^\prime)$ is called a {\em differential field.} An element $c$ of $K$ is called a {\em constant} if $c^\prime=0$. The set of constants in $K$, denoted by $C_K$, is a subfield of~$K$. Set $K^\prime := \{ f^\prime \mid f \in K\}$,  which is a linear subspace over~$C_K$. We call $K^\prime$ the {\em integrable subspace} of $K$.

Let $(E, \, ^\delta)$ be a differential field containing $K$. We say that $E$ is a {\em differential field extension} of~$K$ if $^\delta|_{K}=\,^\prime$. The derivation $^\delta$ is also denoted by $^\prime$ when there is no confusion.
Let $(F, \, ^\delta)$ be another differential field. An algebraic homomorphism $\phi$ from~$K$ to~$F$ is said to be {\em differential} if $\phi(f^\prime)=\phi(f)^\delta$ for all $f \in K$.

Let $(K,\,^\prime)$ be a differential field and $f \in K$. We call $f$ a {\em logarithmic derivative} in $K$ if $f = g^\prime / g$ for some $g \in K$. Let $K(t)$ be a differential extension of $K$ where $t$ is transcendental over $K$ and $t'\in K[t]$.
A polynomial~$p$ in $K[t]$ is said to be {\em $t$-normal} if $\gcd(p, p^\prime)=1$. For $f \in K(t)$, we say that $f$ is {\em $t$-simple} if it is $t$-proper and has a $t$-normal denominator.

We next define primitive and logarithmic generators, which are based on Definitions 5.1.1 and 5.1.2 in~\cite{BronsteinBook}\footnotemark[1], respectively.

\begin{defn} \label{DEF:primitive}
Let $(K, \, ^\prime)$ be a differential field, and $E$ be a differential field extension of $K$. An element $t$ of $E$ is said to be {\em primitive} over $K$ if $t^\prime \in K$. A primitive element $t$ is called a {\em primitive generator} over~$K$ if it is transcendental over~$K$ and $C_{K(t)}=C_K$. 
Furthermore, a primitive generator $t$ is called a {\em logarithmic generator} over $K$ if  $t'$ is a $C$-linear combination of logarithmic derivatives in $K$.
\end{defn}

An immediate consequence of Theorem 5.1.1 in~\cite{BronsteinBook}\footnotemark[1]\ is:

\begin{prop} \label{PROP:primitivemonomial}
Let $t$ be primitive over~$K$. Then $t$ is a primitive generator over $K$ if and only if $t^\prime \notin K^\prime$. Assume that $t$ is a primitive generator over $K$. Then $p \in K[t]$ is $t$-normal if and only if $p$ is squarefree.
\end{prop}

For the rest of the section, assume that $(K, \, ^\prime)$ is a differential field, and that $t$ is a primitive generator over $K$. By Theorem 5.3.1 in~\cite{BronsteinBook}\footnotemark[1] and Lemma 2.1 in~\cite{CDL2018}\footnotemark[1], for each $f \in K(t)$, there exists a unique $t$-simple element $h$ such that
\begin{equation} \label{EQ:hermite}
f \equiv h \mod \big( K(t)^\prime + K[t]\, \big).
\end{equation}
We call $h$ the {\em Hermitian part} of $f$ with respect to $t$, and denote it by $\hp_t(f)$. It is easy to check that $\text{hp}_t$ is a $C_K$-linear map on $K(t)$. Because of the uniqueness of Hermitian parts and Lemma 2.1 in \cite{CDL2018}\footnotemark[1], we have the following lemma.

\begin{lemma} \label{LM:int}
Let $f, g \in K(t)$. Then
\begin{itemize}
\item[(i)] $f \in K(t)^\prime + K[t] \Longrightarrow \hp_t(f)=0$,
\item[(ii)] $f$ is $t$-simple $\Longrightarrow  f = \hp_t(f)$, and
\item[(iii)] $f \equiv g \mod ( K(t)^\prime + K[t]) \Longrightarrow \hp_t(f)= \hp_t(g)$.
\end{itemize}
\end{lemma}

The next two lemmas give some nice properties of proper elements and logarithmic derivatives.

\begin{lemma}  \label{LM:proper}
If $f{\in}K(t)$ is $t$-proper, then $f{-}\hp_t(f){\in}K(t)^\prime$.
\end{lemma}

\begin{proof}
Since $t$ is a primitive generator over $K$, the derivative of a $t$-proper element of $K(t)$ is also $t$-proper. By~\eqref{EQ:hermite}, $f = \hp_t(f) + g^\prime + p$ for some $g \in K(t)$ and $p \in K[t]$. Let $r$ be the $t$-proper part of $g$. Thus, $f - \hp_t(f) - r^\prime = p + (g - r)^\prime$
whose left-hand side is $t$-proper and whose right-hand side is a polynomial in $t$. Thus, both sides must be zero.  Consequently, $f - \hp_t(f)=r^\prime\in K(t)'.$
\end{proof}

\begin{lemma} \label{LM:logder}
Let $f \in K(t)$ be a logarithmic derivative.
\begin{itemize}
\item[(i)] If $f$ is $t$-proper, then $f$ is $t$-simple.
\item[(ii)] There exists a $t$-simple logarithmic derivative $g \in K(t)$ and a logarithmic derivative $h \in K$ such that $f=g+h$.
\end{itemize}
\end{lemma}
\begin{proof}
(i) The only thing we need to show is that the denominator of $f$ is $t$-normal.
By the logarithmic derivative identity \cite[Theorem 3.1.1 (v)]{BronsteinBook}\footnotemark[1], the denominator of $f$ is squarefree,
which is also $t$-normal by Proposition~\ref{PROP:primitivemonomial}.

(ii) By irreducible factorization and the logarithmic derivative identity, $f = \left(\sum_{i} m_i p_i^\prime/p_i\right) + \alpha^\prime/\alpha,$
where $\alpha \in K$, $m_i \in \Z$, and $p_i \in K[t]$ is monic  irreducible and pairwise coprime. Then each $p_i^\prime/p_i$ is $t$-proper, because $t$ is primitive over $K$.
Setting $g=\sum_{i} m_i p_i^\prime/p_i$ and $h=\alpha^\prime/\alpha$ yields (ii). 
\end{proof}

The following lemma will be useful when we construct our remainders. This is the same as Lemma~2.3 in \cite{CDL2018}.

\begin{lemma} \label{LM:lc}
Let $p\in K[t].$ If $p\in K(t)'$, then the leading coefficient of $p$ is equal to $ct'+b'$ for some $c \in C_K$ and $b\in K$. As a special case, if $p \in K \cap K(t)^\prime$, then $p \equiv c t^\prime \mod K'$. 
\end{lemma}

\section{Matryoshka Decompositions}
\label{SECT:md}

We denote $\{1,2,\ldots,n \}$ and $\{0,1,2,\ldots,n \}$ by~$[n]$ and~$[n]_0$, respectively. Let $(K_0,\, ^\prime)$ be a differential field and for each $i \in [n]$, $K_i =K_{i-1}(t_i)$, where $t_i$ is transcendental over~$K_{i-1}$ and $t_i^\prime \in K_i$. Then we have a tower of differential extensions:
\begin{equation} \label{EQ:tower}
\begin{array}{ccccccc}
K_0 & \subset & K_1 & \subset & \cdots & \subset & K_n \\
&  & \shortparallel &  & &   & \shortparallel \\
& & K_0(t_1) & \subset & \cdots & \subset & K_{n-1}(t_n).
\end{array}
\end{equation}
We use $K_0(\bar{t})$ to denote the tower \eqref{EQ:tower}, where $\bar{t}:=(t_1,\ldots,t_n)$ refers to the generators in the chain of field extensions
(to contrast with $K_n$, which is just the largest field in the chain).

We can describe $K_0(\bar t)$ based on the nature of its generators. If $K_0 = (C(x), d/dx)$ and each $t_i$ in~\eqref{EQ:tower} is a primitive generator over~$K_{i-1}$ for all $i \in [n]$, then we call $K_n$ a {\em primitive extension} over~$K_0$ and $K_0(\bar t)$ a {\em primitive tower}. By Definition~\ref{DEF:primitive}, $C_{K_n}=C_{K_0}$, which is equal to $C$. Furthermore, a primitive tower is said to be {\em logarithmic} if each $t_i$ is a logarithmic generator over $K_{i-1}$. For brevity, the primitive tower $K_0(\bar t)$ is also denoted by $K_n$ when its generators are clear from the context.

For each $i \in [n]$, an element of $K_n$ from \eqref{EQ:tower} is said to be {\em $t_i$-proper} if it is free of $t_{i+1}, \ldots, t_{n}$ and the degree of its numerator in $t_i$ is lower than that of its denominator. Denote by $T_i$ the multiplicative monoid generated by $t_{i+1}, \ldots, t_n$ for all~$i$ with $0 \leq i <n$, and set $T_n=\{1\}$. For each $i~\in~[n]$, let~$P_i$ be the additive group consisting of all the linear combinations of the elements of $T_i$ whose coefficients are $t_i$-proper. Furthermore, let $P_0=K_0[t_1,\ldots,t_n]$.  All of the $P_i$'s are closed under multiplication. A routine induction based on~\eqref{EQ:decomposition} shows
\begin{equation} \label{EQ:direct}
K_n= \bigoplus_{i=0}^n P_i.
\end{equation}

Let $\pi_i$ be the projection from $K_n$ onto $P_i$ with respect to~\eqref{EQ:direct}. For every element $f \in K_n$, we have that
$$f = \sum_{i=0}^n \pi_i(f),$$
which is called the {\em matryoshka decomposition} of $f$. 
Figure \ref{FIG:matryoshka} illustrates this namesake. We also call~$\pi_i(f)$ the {\em $i$-th projection} of~$f$ for all $i \in [n]_0$. This new view allows us to describe the following ordering (which will be used to define a remainder).

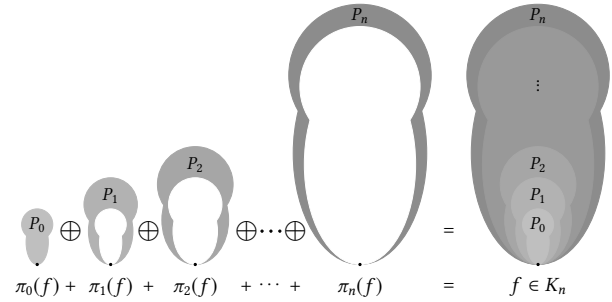
\begin{figure}[ht]
\centering
\resizebox{8cm}{!}{
\begin{tikzpicture}[thick,scale=1, every node/.style={scale=1.5}]
\tikzmath{\xpos=-11; \ypos=0; \dotrad=0.05; \ybtwn=-2.5; \xposzero=-1.5; \xposone=1.8;\xpostwo=5.6; \xposn=13; \xright=21; \yproj=-5; \sizefont=3; \sizefontp=3.5;}
\draw[color=gray!50, fill=gray!50] (\xpos+\xposzero,\ypos-3) ellipse (0.5 and 1);
\draw[color=gray!50, fill=gray!50] (\xpos+\xposzero,\ypos-2.2) circle (0.7cm);
\draw[fill] (\xpos+\xposzero,\ypos-4) circle (\dotrad);
\node at (\xpos+\xposzero,\ypos+\yproj) [fontscale=\sizefontp] {$\pi_0(f)$};
\node at (\xpos+\xposzero+1.5,\ypos+\yproj) [fontscale=\sizefont] {$+$};
\node at (\xpos+\xposzero,\ypos-2.2) [fontscale=\sizefont] {$P_0$};
\node at (\xpos+\xposzero+1.5,\ypos+\ybtwn)[fontscale=\sizefont] {$\bigoplus$};
\draw[color=gray!60, fill=gray!60] (\xpos+\xposone,\ypos-2.5) ellipse (1 and 1.5);
\draw[color=gray!60, fill=gray!60] (\xpos+\xposone,\ypos-1.3) circle (1.2cm);
\draw[color=white, fill=white] (\xpos+\xposone,\ypos-3) ellipse (0.5 and 1);
\draw[color=white, fill=white] (\xpos+\xposone,\ypos-2.2) circle (0.7cm);
\draw[fill] (\xpos+\xposone,\ypos-4) circle (\dotrad);
\node at (\xpos+\xposone,\ypos+\yproj)[fontscale=\sizefontp] {$\pi_1(f)$};
\node at (\xpos+3.5,\ypos+\yproj) [fontscale=\sizefont]{$+$};
\node at (\xpos+\xposone,\ypos-0.9) [fontscale=\sizefont]{$P_1$};
\node at (\xpos+3.5,\ypos+\ybtwn)[fontscale=\sizefont] {$\bigoplus$};
\draw[color=gray!70, fill=gray!70] (\xpos+\xpostwo,\ypos-2) ellipse (1.5 and 2);
\draw[color=gray!70, fill=gray!70] (\xpos+\xpostwo,\ypos-0.4) circle (1.7cm);
\draw[color=white, fill=white] (\xpos+\xpostwo,\ypos-2.5) ellipse (1 and 1.5);
\draw[color=white, fill=white] (\xpos+\xpostwo,\ypos-1.3) circle (1.2cm);
\draw[fill] (\xpos+\xpostwo,\ypos-4) circle (\dotrad);
\node at (\xpos+\xpostwo,\ypos+\yproj)[fontscale=\sizefontp] {$\pi_2(f)$};
\node at (\xpos+\xpostwo,\ypos+0.5) [fontscale=\sizefont]{$P_2$};
\node at (\xpos+7.9,\ypos+\yproj)[fontscale=\sizefont] {$+$};
\node at (\xpos+7.9,\ypos+\ybtwn) [fontscale=\sizefont]{$\bigoplus$};
\draw[fill] (\xpos+8.6,\ypos+\ybtwn) circle (0.05);
\draw[fill] (\xpos+9,\ypos+\ybtwn) circle (0.05);
\node at (\xpos+9,\ypos+\yproj)[fontscale=\sizefont] {$\cdots$};
\draw[fill] (\xpos+9.4,\ypos+\ybtwn) circle (0.05);
\node at (\xpos+10.1,\ypos+\yproj)[fontscale=\sizefont] {$+$};
\node at (\xpos+10.1,\ypos+\ybtwn)[fontscale=\sizefont] {$\bigoplus$};
\draw[color=gray!90, fill=gray!90] (\xpos+\xposn,\ypos+1) ellipse (3 and 5);
\draw[color=gray!90, fill=gray!90] (\xpos+\xposn,\ypos+4.5) circle (3.2cm);
\draw[color=white, fill=white] (\xpos+\xposn,\ypos+0.5) ellipse (2.5 and 4.5);
\draw[color=white, fill=white] (\xpos+\xposn,\ypos+4) circle (2.7cm);
\draw[fill] (\xpos+\xposn,\ypos-4) circle (\dotrad);
\node at (\xpos+\xposn,\ypos+\yproj)[fontscale=\sizefontp] {$\pi_n(f)$};
\node at (\xpos+17,\ypos+\yproj)[fontscale=\sizefont] {$=$};
\node at (\xpos+\xposn,\ypos+7.2)[fontscale=\sizefont] {$P_n$};
\node at (\xpos+17,\ypos+\ybtwn) [fontscale=\sizefont]{$=$};

\draw[color=gray!90, fill=gray!90] (\xpos+\xright,\ypos+1) ellipse (3 and 5);
\draw[color=gray!90, fill=gray!90] (\xpos+\xright,\ypos+4.5) circle (3.2cm);
\draw[color=gray!80, fill=gray!80] (\xpos+\xright,\ypos+0.5) ellipse (2.5 and 4.5);
\draw[color=gray!80, fill=gray!80] (\xpos+\xright,\ypos+4) circle (2.7cm);
\draw[color=gray!70, fill=gray!70] (\xpos+\xright,\ypos-2) ellipse (1.5 and 2);
\draw[color=gray!70, fill=gray!70] (\xpos+\xright,\ypos-0.4) circle (1.7cm);
\draw[color=gray!60, fill=gray!60] (\xpos+\xright,\ypos-2.5) ellipse (1 and 1.5);
\draw[color=gray!60, fill=gray!60] (\xpos+\xright,\ypos-1.3) circle (1.2cm);
\draw[color=gray!50, fill=gray!50] (\xpos+\xright,\ypos-3) ellipse (0.5 and 1);
\draw[color=gray!50, fill=gray!50] (\xpos+\xright,\ypos-2.2) circle (0.7cm);
\draw[fill] (\xpos+\xright,\ypos-4) circle (\dotrad);
\node at (\xpos+\xright,\ypos+7.2)[fontscale=\sizefont] {$P_n$};
\node at (\xpos+\xright,\ypos+4.2)[fontscale=\sizefont] {$\vdots$};
\node at (\xpos+\xright,\ypos+0.5)[fontscale=\sizefont] {$P_2$};
\node at (\xpos+\xright,\ypos-0.9)[fontscale=\sizefont] {$P_1$};
\node at (\xpos+\xright,\ypos-2.2)[fontscale=\sizefont] {$P_0$};
\node at (\xpos+\xright,\ypos+\yproj)[fontscale=\sizefontp] {$f\in K_n$};
\end{tikzpicture}
}

\caption{Matryoshka Decomposition}
\label{FIG:matryoshka}
\Description[short]{long}
\end{figure}

Suppose that $\prec$ is the purely lexicographic order on $T_0$, in which $t_1 \prec t_2 \prec \cdots \prec t_n$. Then $\prec$ is also a monomial order on each $T_i$, because $T_i \subseteq T_0.$ For $f \in K_n$ and $i \in [n]_0$, the $i$-th projection of $f$ can be viewed as a polynomial in $K_i[t_{i+1}, \ldots, t_n]$, which allows us to define the {\em $i$-th head monomial} of $f$, denoted by~$\hmm_i(f)$, to be the highest monomial in $T_i$ that appears in~$\pi_i(f)$ if $\pi_i(f)$ is non-zero, and zero if $\pi_i(f)$ is zero.

We define the {\em $i$-th head coefficient of $f$}, denoted by $\hc_i(f)$, to be the coefficient of $\hmm_i(f)$ in $\pi_i(f)$ if $\pi_i(f)$ is non-zero, and zero if $\pi_i(f)$ is zero. By the matryoshka decomposition, $\hc_i(f)$ is $t_i$-proper for all $i \in [n].$

The {\em head monomial} of~$f$, denoted by~$\hmm(f)$, is defined to be the highest monomial among $\hmm_0(f)$, $\hmm_1(f),$ \ldots, $\hmm_n(f)$, in which zero is regarded as the lowest \lq\lq monomial\rq\rq.
Let $\bI_f = \{ i \in [n]_0 \mid  \hmm_i(f) = \hmm(f) \}.$
The {\em head coefficient} of $f$, denoted by $\hc(f)$, is defined to be $\sum_{i \in \bI_f} \hc_i(f)$.

\begin{defn} \label{DEF:order}
For $f, g \in K_n$, denote $d_f$ and $d_g$ to be the degrees of the  denominators of $f$ and $g$ with respect to $t_n$, respectively. We say that $f$ is {\em lower than} $g$, denoted by $f \prec g$, if either $d_f<d_g$, or $d_f=d_g$ and $\hmm(f) \prec \hmm(g).$ We say that $f$ is {\em not higher than}~$g$, denoted by $f\preceq g$, if either $f \prec g$, or $d_f=d_g$ and $\hmm(f) = \hmm(g).$
\end{defn}

Since $\prec$ on $T_0$ is a Noetherian total order, the partial order on $K_n$ given by Definition \ref{DEF:order} is also Noetherian, that is, every nonempty set in $K_n$ has a minimal element w.r.t.~$\prec$. We can use this order to define a desired remainder of the given function. Let $f \in K_n$ and
\begin{equation} \label{EQ:equiv}
R_f := \{g \in K_n \mid g \equiv f \mod K_n^\prime\}.
\end{equation}
Thus, there exists a minimal element $r \in R_f$. We note that such a minimal element is not unique.

\begin{defn} \label{DEF:remainder}
Given $f\in K_n$, a minimal element of $R_f$ is said to be a {\em remainder} of $f$. Moreover, let $r\in K_n$. Then we say that $r$ is a remainder if $r$ is a remainder of itself.
\end{defn}

As usual, simple elements (or Hermitian parts) play an important role when we construct remainders. Before we move on to the next section, we first generalize the definition of $t$-simple elements from the previous section with the help of the matryoshka decomposition.

\begin{defn} \label{DEF:simpleintower}
An element $f \in K_n$ is said to be {\em simple} if  $\pi_i(f)$ is $t_i$-simple for all $i \in [n]_0$, where $t_0=x$.
\end{defn}

\newpage

\section{Additive Decompositions}
\label{SECT:ad}

Remainders in a tower are described in terms of minimality, which is not constructive. In this section, we will present an algorithm for constructing a remainder in an S-primitive tower (see Definition~\ref{DEF:S-primitive}), based on Hermite reduction and integration by parts. To know when to terminate the algorithm, we need to be able to identify the first generator present in a given monomial  (this is the same notion as {\em scale} in~\cite{CDL2018}).

\begin{defn} \label{DEF:indicator}
For a monomial $M = t_{1}^{d_1} \cdots t_n^{d_n} \in T_0$,
the {\em indicator} of~$M$, denoted by $\ind_n(M)$, is defined to be $n$ if $M=1$,
or defined to be $\min\{i\in[n] \mid d_i\neq 0\}$.
\end{defn}

For $M \in T_0$, we set $K_n^{(\prec M)} := \left\{ f \in K_n
\mid \hmm(f)\prec M \right\}.$ Note that  $K_n^{(\prec M)}$ is closed under addition. The following lemma describes sufficient conditions for reducing a given term with respect to $\prec$ via integration by parts.

\begin{lemma} \label{LM:ibpreduce}
Let $K_n$ be primitive, $M \in T_0$ with indicator~$m$, and $a \in K_{m-1}$.
Then $a M \in K_n^\prime + K_n^{(\prec M)}$ if
\begin{itemize}
\item[(i)] $a \in K_{m-1}^\prime$, or
\item[(ii)] $a\in\spa_C\{t_1^\prime, \ldots, t_m^\prime\}.$
\end{itemize}
\end{lemma}

\begin{proof}
It is obvious for $M=1$. Assume that $M \neq 1$.

(i) Let $M=t_{m}^{d_{m}} \cdots t_n^{d_n}$ for $d_{m}, \ldots, d_n \in \N$ and $d_m >0$.
Since $K_n$ is a primitive extension over $K_0$, we have $t_j^\prime \in K_{j-1}$ for each $j$ with $m\leq j \leq n$. Then
\begin{equation} \label{EQ:derM}
M^\prime = \sum_{j=m}^n h_j N_j,
\end{equation}
where $h_j$ belongs to $K_{j-1}$, and $N_j$ is either equal to zero if $d_j=0$ or $t_j^{d_j-1} t_{j+1}^{d_{j+1}} \cdots t_n^{d_n}$ if $d_j>0$.
There exists $g \in K_{m-1}$ such that $a  = g^\prime$, because $a \in K_{m-1}^\prime$. With integration by parts and~\eqref{EQ:derM}, we see that
$g^\prime M = (g M)^\prime + \sum_{j=m}^n (-g h_j) N_j.$
Let $M_j = t_j^{d_j} t_{j+1}^{d_{j+1}} \cdots t_n^{d_n}$ for all $j$ with $m \le j \le n$. Then  $N_j \prec M_j \prec M$ implies $-g h_j N_j \prec M_j \prec M$ because $g h_j$ is free of $t_j, t_{j+1},\ldots,t_n$, and $N_j\prec M_j$. It follows that $\sum_{j=m}^n (-g h_j) N_j \prec M$ and $a M \in K_n^\prime + K_n^{(\prec M)}$.

(ii) Let $M = t_{m}^d N$, where $d \in \Z^+$ and $N \in T_{m}$. Since $a\in\spa_C\{t_1^\prime, \ldots, t_m^\prime\}$, $a = g + h$, where $g \in K_{m-1}^\prime$ and $h = c t_{m}^\prime$ for some $c \in C$. Then $g M \in  K_n^\prime + K_n^{(\prec M)}$ by (i) and \[ h M = c t_m^\prime t_m^d N = \left(\frac{c}{d+1} t_m^{d+1} \right)^\prime N. \] The lemma holds since $h M \in K_n^\prime + K_n^{(\prec N)}$ and $N \prec M$.
\end{proof}

In order to avoid increasing the order during the process and obtain sufficient {\em and} necessary conditions, we need to impose an extra condition on the generators: $$\hmm(t_i^\prime)=1 \text{ for all } i \in [n].$$
By Lemma~\ref{LM:proper} and the rational additive decomposition, for all $i \in [n]$, there exists a simple $h_i$ in  $K_{i-1}$ and a $g_i \in K_{i-1}$ such that $t_i^\prime = g_i^\prime + h_i$. Let $u_i = t_i-g_i$. Then $u_i$ is a primitive generator over $K_{i-1}$. Moreover, $K_0(\bar t)=K_0(\bar u)$. Therefore, without loss of generality, we can further assume that each $t_i^\prime$ is simple in $K_{i-1}$ for all $i \in [n]$.

\begin{defn}\label{DEF:S-primitive}
A tower $K_0(\bar t)$ is said to be {\em S-primitive} if it is a primitive tower and $t_i^\prime$ is simple for all $i \in [n]$.
\end{defn}

Our next goal is to construct remainders in S-primitive towers based on a special property of simple elements.

\begin{lemma} \label{LM:reduce0}
Let $K_n$ be an S-primitive tower. If $f \in K_n^\prime$ is simple, then $f \in \spa_C\{t_1^\prime, \ldots, t_n^\prime\}.$
\end{lemma}

\begin{proof}
Since $f \in K_n^\prime$ and $\pi_n(f)$ is $t_n$-simple, $\pi_n(f) = \hp_{t_n}(f)=0$ by Lemma~\ref{LM:int} (i) and (ii). Thus, $f \in K_{n-1}$.

We proceed by induction on $n$.
If $n=1$, then $f \in K_0 \cap K_1^\prime$ is $x$-simple by Definition~\ref{DEF:simpleintower}. By Lemma~\ref{LM:lc}, there exists a $c \in C$ such that $f \equiv ct_1^\prime \mod K_0^\prime$. Since both $f$ and $t_1^\prime$ are $x$-simple, we have that $f = ct_1^\prime$ by Lemma~\ref{LM:int} (ii) and (iii).

Assume that $n>1$ and the lemma holds for $n-1$.
For $f$ in $K_{n-1} \cap K_n^\prime$, there is a $c \in C$ such that $f \equiv c t_n^\prime \mod K_{n-1}^\prime$ by Lemma~\ref{LM:lc}. Then $f-ct_n^\prime \in K_{n-1}^\prime$. 
Since both $f$ and $t_n^\prime$ are simple, $f - ct_n^\prime$ is also simple. By the induction hypothesis, we have that $f - ct_n^\prime \in \spa_C\{t_1^\prime, \ldots, t_{n-1}^\prime\}$, which implies that $f \in \spa_C\{t_1^\prime, \ldots, t_n^\prime\}$.
\end{proof}

The previous lemma gives us a direct way to determine whether or not a tower is S-primitive.

\begin{cor} \label{COR:Spri}
The tower $K_n$ is S-primitive if and only if for all $i \in [n]$, $t_i' \in K_{i-1}$ is simple and $t_1', \ldots, t_n'$ are $C$-linearly independent.
\end{cor}
\begin{proof}
If $K_n$ is an S-primitive tower, then $t_i'$ is simple for all $i\in[n]$. Furthermore, $t_i'\notin K_{i-1}'$ for all $i\in[n]$ by Proposition~\ref{PROP:primitivemonomial}. So $t_1',\ldots,t_n'$ are $C$-linearly independent.

We prove the converse by induction. If $n=1$, then a non-zero and simple $t_1'$ clearly implies that $K_1$ is S-primitive.  Suppose $n>1$ and the implication holds for~$n-1$. Assume that for all $i \in [n]$, $t_i' \in K_{i-1}$ is simple and that $t_1', \ldots, t_n'$ are $C$-linearly independent. By the induction hypothesis, $K_{n-1}$ is S-primitive. By Lemma~\ref{LM:reduce0}, $t_n^\prime \notin \spa_C\{t_1', \ldots, t_{n-1}'\}$ implies that $t_n' \notin K_{n-1}^\prime$. Thus, $t_n$ is a primitive generator over $K_{n-1}$ by Proposition~\ref{PROP:primitivemonomial}. Accordingly, $K_n$ is S-primitive.
\end{proof}

The following lemma gives a sufficient and necessary
condition in S-primitive towers for lowering an element  with respect to $\prec$ modulo the integrable space.

\begin{lemma} \label{LM:reduce}
Suppose that $K_n$ is an S-primitive tower. Let $M \in T_0$ with $\ind_n(M)=m$ and $a \in K_{m-1}$ be simple. Then $a M \in K_n^\prime+ K_n^{(\prec M)}$ if and only if  $a \in \spa_C\{t_1^\prime, \ldots, t_m^\prime\}.$
\end{lemma}

\begin{proof}
The sufficiency follows from Lemma~\ref{LM:ibpreduce} (ii). Conversely, assume that $a M \in K_n^\prime+ K_n^{(\prec M)}$. If $M=1$, then $m=n$ and $a \in K_n^\prime$. By Lemma~\ref{LM:reduce0}, $a \in \spa_C\{t_1^\prime, \ldots, t_n^\prime\}.$ If $M \succ 1$ with $M = t_m^{d_m} \cdots t_n^{d_n}$ and $d_m >0$, we can proceed by induction on $n$.

For the base case, $a M\in K_1^\prime + K_1^{(\prec M)}$ implies that there exists a $t_1$-proper element $b \in K_1$ and $p \in K_0[t_1]$ with $\deg_{t_1}(p)<d_1$ such that $a M + b + p \in  K_1^\prime$.
By Lemma~\ref{LM:proper} and Lemma~\ref{LM:int} (i), $a M + p \in K_1^\prime$. Then Lemma~\ref{LM:lc} implies that $a - c t_1^\prime \in K_0^\prime$ for some $c \in C$.
Hence, $a = c t_1^\prime$, because~$a$ and $t_1^\prime$ are both $x$-simple.

\newpage

Assume that $n >1$ and the conclusion holds for $n-1$. Let $N = M/t_n^{d_n}$, which is a power product of $t_m, \ldots, t_{n-1}$.
Since $a M \in K_n^\prime+ K_n^{(\prec M)}$,
there is a $t_n$-proper element $b$ and $p \in K_{n-1}[t_n]$ with $\hmm(p) \prec M$ such that
$a N t_n^{d_n} + b + p \in K_n^\prime.$
By Lemma~\ref{LM:proper}, we can assume that $b$ is $t_n$-simple.
So, $b = 0$ by Lemma~\ref{LM:int} (i).
Let $p = q t_n^{d_n} + r$ such that $q \in K_{n-1}$ with $\hmm(q) \prec N$ and $r \in K_{n-1}[t_n]$ with $\deg_{t_n}(r) < d_n$.
Then we have $(a N + q) t_n^{d_n} + r \in K_n^\prime.$
By Lemma~\ref{LM:lc}, there exists $c \in C$ such that $a N + q - c t_n^\prime \in K_{n-1}^\prime$. Hence,
\begin{equation} \label{EQ:reduce1}
a N \equiv  c t_n^\prime \mod \big( K_{n-1}^\prime + K_{n-1}^{(\prec N)} \, \big).
\end{equation}
If $N=1$, then $m=n$ and $a \in K_n^\prime$.
By Lemma~\ref{LM:reduce0}, we have that $a \in \spa_C\{t_1^\prime, \ldots, t_n^\prime\}$. The lemma holds. If $N \succ 1$, then $\ind_{n-1}(N) = m <n$. By~\eqref{EQ:reduce1}, $aN \in K_{n-1}^\prime + K_{n-1}^{(\prec N)}$,
because $\hmm(c t_n^\prime)=1$. It follows from the the induction hypothesis that $a \in \spa_C\{t_1^\prime, \ldots, t_m^\prime\}$.
\end{proof}

We can now specify a remainder in S-primitive towers and prove that the algorithm to construct it will terminate.

\begin{prop} \label{PROP:remainder}
Let $K_n$ be an S-primitive tower, and $r \in K_n$ with $m = \ind_n(\hmm(r))$. Then $r$ is a  remainder if either $r=0$,
or $\pi_n(r)$ is $t_n$-simple and $\hc(r-\pi_n(r))$ is simple and is not a nonzero element of $\spa_C\{t_1^\prime, \ldots, t_m^\prime\}$.
\end{prop}
\begin{proof}
Let $f \in R_r$  as defined in \eqref{EQ:equiv}.
As $\pi_n(r)$ is $t_n$-simple, we have $\hp_{t_n}(f) = \pi_n(r)$ by Lemma~\ref{LM:int} (ii) and (iii). Then the denominator of $\pi_n(r)$, which is exactly the denominator of $r$ as a polynomial in $K_{n-1}[t_n]$, divides the denominator of $f$ by Theorem 5.3.1 in~\cite{BronsteinBook}\footnotemark[1].

We further need to  show that $\hmm(r) \preceq \hmm(f)$.  Suppose the contrary. Then $r \neq 0$. Let $M = \hmm(r)$ and $a = \hc(r-\pi_n(r))$.

If $M=1$, then $m = n$, $a=r-\pi_n(r)$, and $f = 0$, which implies that $r \in  K_n^\prime$. Then $\pi_n(r)=0$ by Lemma~\ref{LM:int}~(i). So, $a \in K_{n-1} \cap K_n^\prime$. By Lemma~\ref{LM:reduce0}, we have that $a$ belongs to $\spa_C\{t_1^\prime, \ldots, t_n^\prime\}$. Thus, $a=r=0$, a contradiction.

Assume that $M \succ 1$.  Since $M\succ \hmm(f)$, we have that $\hmm(r-f)=M \text{ and } \hc(r-f)=\hc(r).$ Then $\hc(r-f)=a$ because $M\succ 1$ and $\hmm(\pi_n(r))=1$. From $r-f \in K_n^\prime$, we see that $a \, M \in K_n^\prime+ K_n^{(\prec M)}.$ By Lemma~\ref{LM:reduce}, $a$ belongs to  $\spa_C\{t_1^\prime, \ldots, t_m^\prime\}$,
which implies that $a=0$. Then $r=\pi_n(r)$ and $M=1$, a contradiction.
\end{proof}

\begin{theorem}\label{TH:existence}
Let $K_n$ be an S-primitive tower and let $f \in K_n$. Then one can construct a remainder of $f$ with the properties described in Prop.~\ref{PROP:remainder} in a finite number of steps.
\end{theorem}

\begin{proof}
By Lemma~\ref{LM:proper}, $\pi_n(f) \equiv \hp_{t_n}(f) \mod K_n^\prime$. Then
\begin{equation} \label{EQ:first}
f \equiv \hp_{t_n}(f) + (f-\pi_n(f)) \mod K_n^\prime.
\end{equation}
The $n$-th projection of the right-hand side of the congruence is equal to $\hp_{t_n}(f)$, which is $t_n$-simple.

Let $M = \hmm(f -\pi_n(f))$.
We proceed by a Noetherian induction on $M$ with respect to $\prec$. If $M=0$, then $f = \pi_n(f)$. By~\eqref{EQ:first} and Proposition~\ref{PROP:remainder}, $\hp_{t_n}(f)$ is a remainder of $f$.

Assume that $M \neq  0$, and for any $g \in K_n$ with $\hmm(g) \prec M$, there is a remainder $\tilde r$ of $g$ as described in Proposition~\ref{PROP:remainder}.

Let $a = \hc(f -\pi_n(f))$ and $m = \ind_n(M)$.
Since $a \in K_{m-1}$,
its $j$-th projection is equal to zero for each $j \in \{m, \ldots, n\}$.
By Lemma~\ref{LM:proper}, $\pi_i(a) \equiv h_i \mod K_i^\prime$
for some $t_i$-simple elements $h_i \in K_i$ for all $i \in [m-1]_0$ with $t_0=x$. By Lemma~\ref{LM:ibpreduce} (i),
\begin{equation} \label{EQ:ibp}
f-\pi_n(f) \equiv b  M \mod (K_n^\prime+ K_n^{(\prec M)}),
\end{equation}
where $b = \sum_{i=0}^{m-1} h_i$. Note that $b$ is simple by Definition~\ref{DEF:simpleintower}.

If $b \in \spa_C\{t_1^\prime, \ldots, t_m^\prime\}$, then $b M$ is in $K_n^\prime+ K_n^{(\prec M)}$  by Lemma~\ref{LM:ibpreduce} (ii). So $f - \pi_n(f) \equiv g \mod K_n^\prime $ for some $g$ in $K_n^{(\prec M)}$ by~\eqref{EQ:ibp}. Accordingly, $g$ has a remainder~$\tilde r$ as described in Proposition~\ref{PROP:remainder} by the induction hypothesis. It follows that $\hp_{t_n}(f) + \tilde r$ is a remainder of $f$.

Assume that $b  \notin  \spa_C\{t_1^\prime, \ldots, t_m^\prime\}$.  It follows from~\eqref{EQ:first} and~\eqref{EQ:ibp} that
$f \equiv \hp_{t_n}(f) + b M + g \mod K_n^\prime$
for some $g$ in $K_n^{(\prec M)}$. Moreover, we may further assume that $\pi_n(g)$ is $t_n$-simple by Lemma~\ref{LM:proper}. The right-hand side of the above congruence is a remainder as described in Proposition~\ref{PROP:remainder}, because $b$ is the head coefficient of $b M +(g- \pi_n(g))$.
\end{proof}

We now present an algorithm to decompose an element in an S-primitive tower over $K_0=(C(x),d/dx)$ into a sum of a derivative and a remainder. The algorithm is a slight refinement of the proof of the above theorem. We refer the reader to the online supplementary material \footnote{\url{https://wongey.github.io/add-decomp-sprimitive/}} for the implementation.

\vspace{0.3cm}

\noindent\fbox{%
    \parbox{0.47\textwidth}{%
\textsc{AddDecompInField}$\big(\, f,\, K_0(\bar t)\, \big)$

\noindent \textit{Input:} An S-primitive tower $K_0(\bar t)$, described as a list $$\{x, \{t_1, \ldots, t_n\}, \{t_1', \ldots, t_n'\}\},$$ s.t. $t_i^\prime \in K_{i-1}$ is simple for all $i \in [n]$, and $f \in K_n$.

\noindent\textit{Output:} Two elements $g, r \in K_n$ such that $f = g^\prime + r$ and $r$ \\ \qquad satisfies the conditions in Proposition~\ref{PROP:remainder}.
\begin{enumerate}
\item If $f=0$, then return $(0,0)$.
\item Initialize: $M\leftarrow\hmm(f)$, $a\leftarrow\hc(f)$, $m\leftarrow\ind_n(M)$, \\ $d\leftarrow\dg_{t_m}(M)$, $B \leftarrow 0$, $H \leftarrow 0$, $\tilde c \leftarrow 0$.
\item Let $a=\sum_{i=0}^m a_i$ be the matryoshka decomposition.
\item Reduction: For all $i$ from $0$ to $m$, compute $b_i,h_i\in K_i$ s.t. $a_i=b_i'+h_i$, where $h_i$ is $t_i$-simple. Decide whether $\exists\, c_1,\ldots,c_m\in C$ s.t. $h_i=\sum_{j=1}^m c_j t_j'$.
\begin{itemize}
 \item[] Yes: $B\leftarrow B+b_i+\sum_{j=1}^{m-1}c_jt_j$ and $\tilde c\leftarrow\tilde c+c_m;$
 \item[] No: $B\leftarrow B+b_i$ and $H \leftarrow H+h_i$.
\end{itemize}
\item Lower term: $\ell\leftarrow f-aM-BM'-\frac{\tilde c}{d+1}\cdot t_m^{d+1}\cdot \big(M/t_m^d\big)'$\\
Recursion:
$\{\tilde{g},\tilde{r}\}\leftarrow$\textsc{AddDecompInField}$\big(\, \ell,\, K_0(\bar t)\, \big)$
\item Return $g=B M+\frac{\tilde c}{d+1}\cdot t_m \cdot M + \tilde{g}$ and $r=H\cdot M+\tilde{r}$.
\end{enumerate}
    }%
}

\begin{ex} \label{EX:infield}
Find an additive decomposition for
$$f=\frac{1}{\log(x) \Li(x)}+\frac{\Li(x)-2 x \log(x)}{(\log(x))^2}+\log (\log (x)).$$
Then $f$ belongs to the S-primitive tower $$K_3=C(x)(\underbrace{\log(x)}_{t_1},\underbrace{\Li(x)}_{t_2},\underbrace{\log(\log(x))}_{t_3}),$$
 and we can write $f=1/(t_1 t_2)+(t_2-2x t_1)/t_1^2+t_3 \in K_3$. By the above algorithm, we have that
\begin{equation} \label{EQ:ad}
f = \bigg(x t_3+\frac{t_2^2}{2}-t_2-\frac{xt_2+x^2}{t_1}\bigg)^\prime + \underbrace{ \frac{1}{t_1 t_2}}_r.
\end{equation}
The nonzero remainder $r$ implies that $f$ has no integral in~$K_3$.
\end{ex}

An element $f \in K$ is said to have an {\em elementary integral over~$K$} if there exists an elementary extension~$E$ of~$K$ and an element $g$ of $E$ such that $f = g^\prime$ (see \cite[Definition 5.1.4]{BronsteinBook}\footnotemark[1]). We can use the remainder from Theorem \ref{TH:existence} to determine whether or not a function has an elementary integral.

\begin{theorem} \label{TH:ele}
Let $K_n$ be S-primitive and $C$ be algebraically closed.
Let $f \in K_n$ have a remainder $r$ as described in Proposition~\ref{PROP:remainder}.
Then $f$ has an elementary integral over~$K_n$ if and only if
\begin{equation}\label{EQ:elemint}
r \in \spa_C\{t_1^\prime, \ldots, t_n^\prime\} + \spa_C \{g^\prime/g \mid g \in K_n\}.
\end{equation}
\end{theorem}
\begin{proof}
The sufficiency is obvious. Conversely, there exists an $h \in  \spa_C\{g^\prime/g \mid g \in K_n\}$ such that $f \equiv h\mod K_n^\prime$ by Liouville's Theorem~\cite[Theorem 5.5.2]{BronsteinBook}\footnotemark[1]. Since $r$ is a remainder of $f$, we have that $h \equiv  r \mod K_n^\prime$. By Proposition~\ref{PROP:remainder} and Lemma~\ref{LM:logder}, we know that $\pi_n(r)$ and $\pi_n(h)$ are $t_n$-simple, which, together with Lemma~\ref{LM:int} (ii) and (iii),  implies that $\pi_n(r) = \pi_n(h)$. Since $\hmm(h)=1$, we have that $\hmm(r) \preceq 1$ by Definition~\ref{DEF:remainder}. If $\hmm(r) = 0$, then $r = 0$. Otherwise, $\hmm(r) = 1$. By Proposition~\ref{PROP:remainder}, $r$ is simple. Since $h$ is simple, $r-h \in K_n^\prime$ is also simple. By Lemma~\ref{LM:reduce0}, $r-h \in \spa_C\{t_1^\prime, \ldots, t_n^\prime\}$, which implies \eqref{EQ:elemint}.
\end{proof}

\begin{ex}
Let us reconsider the function $f$ and the tower $K_3$ in Example~\ref{EX:infield} under the assumption that $C$ is algebraically closed.
The remainder is $r = t_2^\prime/t_2$. By Theorem~\ref{TH:ele},
$f$ has an elementary integral over $K_3$. It follows from~\eqref{EQ:ad} that
\begin{align*}
\int f \, dx = &\ x \log(\log(x)) + \frac{\Li(x)^2}{2} - \Li(x)-\frac{x\Li(x)+x^2}{\log(x)} \\
            &\ + \log(\Li(x)).
\end{align*}
The Mathematica implementation by Raab based on work in~\cite{Raab2012} computes the same result.
But the ``int(\,)'' command in Maple and the ``Integrate[\, ]'' command in Mathematica both leave the integral unevaluated.
\end{ex}

\section{Logarithmic Towers}
\label{SECT:wglt}

A repeated use of Lemma~\ref{LM:logder} (ii) easily reveals a logarithmic tower to be S-primitive. Hence, \textsc{AddDecompInField} can be applied to all logarithmic towers. In this section, we show that a logarithmic tower can be differentially embedded into a logarithmic tower that we will term ``well-generated'' (see Definition \ref{DEF:wg}) with the aid of the logarithmic derivative identity and the matryoshka decomposition. An element in the latter tower may have a ``finer'' remainder. The logarithmic derivative identity is actually a differential version of logarithmic product and quotient rules, while the matryoshka decomposition
guides us how to apply the rules appropriately.

\begin{ex}
Consider the following function in $x$: $$f=\frac{\log((x+1)\log(x))}{x\log(x)}.$$
For this function, there are two possible ways to construct the tower over $\Q(x)$ containing $f$:
\begin{itemize}
\item[(i)] $t_1=\log(x), t_2=\log((x+1)\ t_1); \ f=\frac{t_2}{xt_1}$,
\item[(ii)] $u_1=\log(x), u_2=\log(x+1), u_3=\log(u_1); \  f=\frac{u_2+u_3}{x u_1}$.
\end{itemize}
In the first tower, $f$ is already a remainder by Proposition~\ref{PROP:remainder}. In the second tower, \textsc{AddDecompInField} computes a remainder $u_2/(xu_1)$ that is lower than $f$. This is because we can decompose
$\log((x+1)\log(x))$ as a sum of $\log(x+1)$ and $\log(\log(x))$ in the second tower, but neither of the two summands is contained in the first.
\end{ex}

We can use the matryoshka decomposition to describe a primitive tower in terms of a matrix, which will be used to rearrange our generators in an order that would yield a finer remainder by applying \textsc{AddDecompInField}.

\begin{defn}
Let $K_0(\bar{t})$ be primitive. The $n \times n$ matrix $$A = \left(  \pi_i(t_j^\prime)  \right)_{0 \le i \le n-1, 1 \le j \le n}$$ is called the {\em matrix associated} to $K_0(\bar{t})$.
\end{defn}

\begin{figure}[ht]
 \vspace{-0.5cm}
\centering
\[
\renewcommand\arraystretch{1.3}
\begin{blockarray}{crcccccc}
&&& t_1' & t_2' & \cdots & t_n' &  \\
&&& \downarrow & \downarrow & & \downarrow \\
\begin{block}{rc(c@{}cccc@{}c)}
P_0  & \to && \star & \star & \cdots & \star &\vphantom{\smash[b]{\bigg|}} \\
P_1 & \to &&  & \star & \cdots & \star & \\
\vdots  &&&  &  & \ddots & \vdots & \\
P_{n-1} & \to &&  & & & \star & \\
\end{block}
\end{blockarray}
\]
 \vspace{-0.5cm}

\caption{A labeled associated matrix of a primitive tower. The $\star$ represents a possibly nonzero element.}
\Description[short]{long}
\label{FIG:assmat0}
\end{figure}

The associated matrix records all information about the derivation on $K_0(\bar t)$, because $\pi_n(t_1^\prime) = \cdots = \pi_n(t_n^\prime) = 0.$  Since $t_j^\prime \in K_{j-1}$ for all $j \in [n]$,  the associated matrix $A$ is in upper triangular form as in Figure~\ref{FIG:assmat0}. Furthermore, if $K_0(\bar t)$ is a logarithmic tower, then the entries of $A$ are all logarithmic derivatives by Lemma~\ref{LM:logder} (ii).

For the following discussion, we will invoke the superscript notation to distinguish between different sets of generators (for example, $\pi_i^{\bar{t}}$ for projections in $K_0(\bar{t})$).

\begin{defn}\label{DEF:si}
 Let $K_0(\bar t)$ be primitive and $f \in K_n \setminus \{0\}$. The {\em significant index} of $f$ is $$\si^{\bar{t}}(f):=\max\{i\in[n]_0 \mid \pi_i(f)\neq 0\}.$$ The vector
$$\sv(\bar t):=\left( \si^{\bar t}(t_1^\prime), \ldots, \si^{\bar t}(t_n^\prime)\right)$$ is called the {\em significant vector} of $K_0(\bar t)$.
Suppose $\sv(\bar t)$ is equal to $(k_1, \ldots, k_n).$ The sequence
$$\scc(\bar t):=\left(\pi^{\bar t}_{k_1}(t_1^\prime), \ldots, \pi^{\bar t}_{k_n}(t_n^\prime)\right)$$ is called the {\em the significant component sequence} of $K_0(\bar{t})$.
\end{defn}

\newpage

The significant vector and significant component sequence are unique with respect to the generators by the matryoshka decomposition.

\begin{ex}
Consider the field $$C(x)\left(\log(x),\log(\log(x)),\log((x+1)\log(x))\right).$$ We set
$t_1=\log(x), t_2=\log(t_1), \text{and } t_3=\log((x+1)\, t_1).$ Then $C(x)(t_1, t_2, t_3)$ is a logarithmic tower whose significant vector is equal to
$(0,1,1)$ and whose significant component sequence is \\ $(1/x, 1/(xt_1), 1/(x t_1))$.

\end{ex}

\begin{defn}\label{DEF:wg}
A logarithmic tower $K_0(\bar t)$ is said to be {\em well-generated} if
\begin{itemize}
\item[(CLI)] $\scc(\bar t)$ is $C$-linearly independent,
\item[(MI)] $\sv(\bar t)$ is (weakly) monotonically increasing, and
\item[(ONE)] each column of its associated matrix contains exactly one non-zero element.
\end{itemize}
\end{defn}

\begin{figure}[ht]
 \vspace{-0.5cm}
\[\left(
\begin{tabular}{cccccccccc}
$\bullet$ & $\cdots$ & $\bullet$ & & & & & & & \\
& & & $\bullet$ & $\cdots$ & $\bullet$ & & & & \\
& & & & & & $\ddots$ & & & \\
& & & & & & & $\bullet$ & $\cdots$  & $\bullet$ \\
& & & & & & & & & \\
\end{tabular}
\right)
\]
 \vspace{-0.5cm}
\caption{The associated matrix of a well-generated tower is in the form of a \lq\lq staircase\rq\rq\, where the $\bullet$'s are $C$-linearly independent and other entries are zero.}
\Description[short]{long}
\label{FIG:WGPT}
\end{figure}

We will show that a logarithmic tower $K_0(\bar t)$ can be embedded into a well-generated one. To this end, we impose the usual lexicographical order on two significant vectors~\cite[Chapter 2, Definition 3]{CLO}\footnotemark[1].

\begin{theorem} \label{TH:wg}
Let $K_0(\bar t)$ be a logarithmic tower.
Then there exists a well-generated logarithmic tower $K_0(\bar u)$,  where $\bar u = (u_1,\ldots,u_w)$ and $n \leq w \leq n(n+1)/2,$  and  a differential  homomorphism $\phi$ from $K_0(\bar t)$ into  $K_0(\bar u )$ with $\phi|_{K_0} = {\rm id}_{K_0}$.
\end{theorem}

\begin{proof}

This proof will be separated into two parts. The first part will show that each primitive (specifically, logarithmic) tower is isomorphic to one where properties (CLI) and (MI) are satisfied. This will enable us to embed the resulting logarithmic tower into a well-generated one, which makes up the second part of the proof.


If $K_0(\bar t)$ does not satisfy (CLI) and (MI), then we can show there exists $v_1,\ldots,v_n\in K_n$ such that $K_0(\bar v)$ is primitive, $K_0(\bar v)=K_0(\bar t)$, and  $\sv(\bar{v})$ is lower than $\sv(\bar{t})$. Since the order of the significant vectors is Noetherian, we can eventually reach a primitive tower that satsifies both (CLI) and (MI).

We start by supposing that $\scc(\bar t)$ is $C$-linearly dependent. Since $\si^{\bar t}(t_1^\prime)=0$,
there exists an $i \in \{2, \ldots, n\}$ and constants $c_1,\ldots,c_{i-1}$ such that
$\scc_i=\sum_{j=1}^{i-1}c_j \cdot \scc_j,$
where~$\scc_j$ is the $j$-th element in $\scc(\bar t)$.
We remove the last non-zero projection of $t_i'$ by setting $v_k:=t_k$ for all $k\in[n]\setminus\{i\}$ and $v_i:=t_i-\sum_{j=1}^{i-1}c_jt_j.$
Thus, $K_0(\bar v)=K_0(\bar t)$. Also,
$\si^{\bar v}(v_k^\prime)=\si^{\bar t}(t_k^\prime)$ for all $k$ in $[n]\setminus\{i\}$ and $\si^{\bar v}(v_i^\prime) < \si^{\bar t}(t_i^\prime).$ We conclude that $K_0(\bar v)$ is a primitive tower with a lower significant vector than $K_0(\bar t)$.

Next, we assume that $\sv(\bar t)$ is not monotonically increasing.
Then there exist an $i \in [n]$ such that
$\si^{\bar t}(t_1')\leq \cdots \leq \si^{\bar t}(t_i')$ and $\si^{\bar t}(t_{i+1}')<\si^{\bar t}(t_i').$ We switch the $i$-th and $(i+1)$-st generators by setting $v_k:=t_k$ for all $k\in[n]\setminus\{i,i+1\}$ and $$v_i:=t_{i+1}; \ v_{i+1}:= t_{i}.$$ Thus, $K_0(\bar v)=K_0(\bar t)$.
Also, $\si^{\bar v}(v_j^\prime) = \si^{\bar t}(t_j^\prime)$ for $j\in[i-1]$ and $\si^{\bar v}(v_i^\prime) < \si^{\bar t}(t_i^\prime)$. Thus, $K_0(\bar v)$ is a primitive tower with a lower significant vector than $K_0(\bar t)$.


If the original primitive tower from the argument is  logarithmic, then the new generators from the above process are  also logarithmic generators. This implies the new tower must be  logarithmic satisfying (CLI) and (MI), and this is what we assume about $K_0(\bar{t})$ from this point forward.


For the second part of the proof, we show that $K_0(\bar{t})$ can be embedded into a well-generated tower. We find the $C$-basis of the associated matrix $\left(\pi_i(t_j')\right)$ by letting $b_1=\pi_0(t_1')$ and identifying all $C$-linearly independent elements $b_2,\ldots, b_w$, ordered by searching the matrix from left to right and top to bottom. Since $K_0(\bar t)$ is primitive, $n\leq w \leq n(n+1)/2$. Since $K_0(\bar{t})$ satisfies (CLI) and (MI), there exist $\ell_1,\ldots,\ell_n \in [w]$ such that $\ell_1=1, \ell_n=w$,
\begin{equation}\label{EQ:sc}
\ell_1<\ell_2<\cdots<\ell_n \ \mbox{and} \  \left(b_{\ell_1},\ldots,b_{\ell_n}\right) =\scc(\bar t).
\end{equation}
By the definition of the associated matrix and the ordering of \\ $\{b_1,\ldots,b_w\}$, for all $j \in [n]$ there exist $c_{j,k}\in C$ such that
\begin{equation} \label{EQ:der}
t_j'=b_{\ell_j}+\sum\limits_{k=1}^{\ell_j-1}c_{j,k}\cdot b_k.
\end{equation}
Let $u_1, \ldots, u_w$ be algebraically independent indeterminates over $K_0$, and $\bar u := (u_1, \ldots, u_w)$. Let $v_j := u_{\ell_j} + \sum_{k=1}^{\ell_j-1} c_{j,k} \cdot u_k$ for all $j\in[n]$.
Then $v_1, \ldots, v_n$ are algebraically independent over $K_0$, because $u_{\ell_j}$ does not appear in the expressions defining $v_1,$ \ldots, $v_{j-1}$. It follows that $\phi:K_0(\bar t) \rightarrow K_0(\bar u)$ defined by $f(t_1, \ldots, t_n)\mapsto f(v_1, \ldots, v_n)$ is a monomorphism and $\phi|_{K_0} = {\rm id}_{K_0}$. For every $k \in [w]$, we define
\begin{equation} \label{EQ:der1}
u_k^\prime = \phi(b_k).
\end{equation}
 Since $u_1, \ldots, u_w$ are algebraically independent over $K_0$, the tower $K_0(\bar u)$ is a differential field
 by Corollary $1^\prime$ in~\cite[page 124]{ZariskiSamuel}\footnotemark[1]. By~\eqref{EQ:der}, $\phi(t_j^\prime)=v_j^\prime$ for all $j \in [n]$.  Thus, $\phi$ is a differential monomorphism.

Lastly, we show that $K_0(\bar{u})$ is a well-generated tower over~$K_0$. Set $\ell_0=0$. For each $k \in [w]$, there exists a $j \in [n]$ such that $\ell_{j-1} < k \le \ell_j$.
Then $s:=\si^{\bar t}(b_k) \le \si^{\bar t}(t_j') <j$ and $b_k$ is $t_s$-proper. Since $\phi$ is a monomorphism, it preserves degrees. By~\eqref{EQ:der1},  $u_k^\prime$ is $u_{\ell_s}$-proper, where $\ell_s \le \ell_{j-1} < k$ since $s<j$.
Hence, $u_k^\prime \in K_0(u_1, \ldots, u_{k-1})$.
Since $\phi$ is differential and $b_k$ is a logarithmic derivative,  $u_k^\prime$ is also a logarithmic derivative by~\eqref{EQ:der1}. In particular, $u_k^\prime$ is $u_{\ell_s}$-simple by Lemma~\ref{LM:logder} (i). Moreover, $b_1, \ldots, b_w$ are $C$-linearly independent, and so are $\phi(b_1), \ldots, \phi(b_w)$ because $\phi$ is a monomorphism. It follows from~\eqref{EQ:der1} that $u_1^\prime, \ldots, u_w^\prime$ are $C$-linearly independent, which implies that $K_0(\bar u)$ is a logarithmic tower by Corollary~\ref{COR:Spri}. In addition, $\pi_i(u_k^\prime)=0$ for all $k \in [w]$ and $i \in [w] \setminus \{\ell_s\}$, because $u_k'$ is $u_{\ell_s}$-proper. Consequently, $K_0(\bar u)$ is well-generated.
\end{proof}

\newpage

The proof of this theorem shows that a logarithmic tower~$\cF$ can be algorithmically embedded in a well-generated tower~$\cE$ by a differential homomorphism $\phi$.
Let $f$ be an element of $\cF$ with a remainder $r$.
Our additive decomposition can be applied to $\phi(f)$ in $\cE$ to get a remainder whose order is not higher than that of $\phi(r)$, and this is what we mean by ``finer''.

The next example illustrates the results of the embedding algorithm and $\textsc{AddDecompInField}$ in both towers.
\begin{ex}
Consider the logarithmic tower
$$\cF = C(x)\Big(\underbrace{\log(x)}_{t_1},\underbrace{\log(xt_1)}_{t_2},\underbrace{\log\big((x+1)(t_1+1)\log(x t_1)\big)}_{t_3}\Big).$$
By Theorem \ref{TH:wg}, there exists a well-generated tower
$$\cE = C(x)\Big(\underbrace{\log(x)}_{u_1},\underbrace{\log(x{+}1)}_{u_2},\underbrace{\log(u_1)}_{u_3},\underbrace{\log(u_1{+}1)}_{u_4},\underbrace{\log( u_1{+}u_3)}_{u_5}\Big)$$
and a differential homomorphism $\phi$ from  $\cF$ to $\cE$ given by
$\phi(t_1) = u_1$, $\phi(t_2)=u_1+u_3$ and $\phi(t_3)=u_2+u_4 + u_5$.
The associated matrices of $\cF$ and $\cE$ are, respectively,
$$\begin{pmatrix}
\frac{1}{x} & \frac{1}{x} & \frac{1}{x+1}\\[6pt]
0 & \frac{t_1'}{t_1} & \frac{t_1'}{t_1+1}\\[6pt]
0 & 0 & \frac{1+t_1}{x t_1t_2}
\end{pmatrix} \,\mbox{and} \,\begin{pmatrix}
\frac{1}{x} & \frac{1}{x+1}& 0 & 0 & 0 \\
0 & 0 & \frac{u_1'}{u_1} & \frac{u_1'}{u_1+1} &0  \\
0 & 0 & 0 & 0 & 0 \\
0 & 0 & 0 & 0 & \frac{(u_1+u_3)'}{u_1+u_3}\\
0 & 0 & 0 & 0 & 0
\end{pmatrix}.$$
\label{EX:logwellgen}

\noindent Let
$$f_1 = \frac{(t_1+1)^2+t_1 t_2}{x t_1 (t_1+1) t_2} \text{ and } f_2 = \frac{t_3}{x}$$
be two elements of $\cF$. Then $\phi(f_1)$ and $\phi(f_2)$ are
$$\frac{(u_1+1)^2+u_1(u_1+u_3)}{x u_1(u_1+1)(u_1+u_3)}  \quad \mbox{and} \quad \frac{u_2+u_4+u_5}{x},$$
respectively. Using \textsc{AddDecompInField}, we compute the respective remainders of $f_1$ and $f_2$ to obtain
$$r_1 = f_1 \quad \text{and}\quad r_2 = \frac{t_1}{-(x+1)}+\frac{1}{x(t_1+1)}+\frac{-(t_1+1)}{x t_2}.$$
In the same vein, we get the remainders of $\phi(f_1)$ and $\phi(f_2)$,
$$\tilde{r}_1=0 \quad \text{and}\quad \tilde{r}_2 = \frac{u_1}{-(x+1)}+\frac{-(u_1+1)}{x(u_1+u_3)},$$
respectively.
Note that $\phi(r_1) \neq 0$ but $\tilde{r}_1=0$, which implies that $\tilde r_1 \prec \phi(r_1)$. While $\tilde{r}_2$ and $\phi(r_2)$ have the same order, we observe that $\tilde r_2$ has fewer nonzero projections than $\phi(r_2)$. 
\end{ex}

\section{Conclusions} \label{SECT:conc}
In this article, we have introduced the matryoshka decomposition to develop an additive decomposition in an S-primitive tower. The decomposition algorithm is based on Hermite reduction and integration by parts. It provides an alternative method for determining in-field (resp.\ elementary) integrability in (resp.\ over) an S-primitive tower without solving any differential equations. Moreover, we embed a logarithmic tower into a well-generated one. The embedding enables us to compute finer remainders.

We observe that the notion of remainders is defined according to a partial order among multivariate rational functions. It would be possible to refine this notion so that remainders possess  certain uniqueness. Moreover, we plan to investigate whether our additive decomposition is applicable to compute telescopers for elements in an S-primitive tower, as carried out in~\cite{CDL2018}. We also hope to develop an additive decomposition in exponential extensions.

\begin{acks}
We are grateful to Shaoshi Chen, Christoph Koutschan and Clemens Raab for their valuable comments and suggestions. H.\ Du and E.\ Wong were supported by the Austrian Science Fund (FWF): F5011-N15. J.\ Guo and Z.\ Li were supported by two NFSC Grants 11688101 and 11771433.
\end{acks}


\newpage
\appendix

\section{Appendix}

For the convenience of the reviewers, this section lists definitions, a lemma, some theorems and a corollary that we use from other books and papers but did not explicitly state in this paper. It will not appear in a formal publication.


\begin{defn}(Definition 5.1.1 in \cite{BronsteinBook})
Suppose $k$ is a differential field and $K$ is a differential extension of $k$. We say that
\begin{itemize}
\item[(i)] $t\in K$ is a {\em primitive over} $k$ if $Dt \in k$,
\item[(ii)] $t \in K^*$ is a {\em hyperexponential over} k if $Dt/t \in k$, and
\item[(iii)] $t\in K$ is {\em Liouvillian over} $k$ if $t$ is either algebraic, a primitive, or a hyperexponential over $k$.
\end{itemize}
$K$ is a {\em Liouvillian
extension} of $k$ if there are $t_1,\ldots,t_{n}$ in $K$ such that $K=k(t_1,\ldots,t_n)$ and $t_i$ is Liouvillian over $k(t_1,\ldots,t_{i-1})$ for $i\in\{ 1 , \ldots , n\}$.
\end{defn}

\begin{defn}(Definition 5.1.2 in \cite{BronsteinBook})
Suppose $k$ is a differential field and $K$ is a differential extension of $k$. We say that $t\in K$ is a {\em Liouvillian monomial over} $k$ if $t$ is transcendental and Liouvillian over $k$ and $C_{k(t)}=C_k$.
\end{defn}

\begin{defn}(Definition 5.1.3 in \cite{BronsteinBook})
$t\in K$ is a {\em logarithm over} $k$ if $Dt=Db/b$ for some $b\in k^*$. $t\in K^*$ is an {\em exponential over} $k$ if $Dt/t=Db$ for some $b\in k$. $t\in K$ is {\em elementary over} $k$ if $t$ is either algebraic, or a logarithm or an exponential over $k$. $t\in K$ is an {\em elementary monomial over} $k$ if $t$ is transcendental and elementary over $k$, and Const$(k(t))=$ Const$(k)$.
\end{defn}

\begin{defn}(Definition 5.1.4 in \cite{BronsteinBook})
$K$ is an {\em elementary extension} of $k$ if there are $t_1,\ldots, t_n$ in $K$ such that $K=k(t_1,\ldots,t_n)$ and $t_i$ is elementary over $k(t_1,\ldots,t_{i-1})$ for $i$ in $\{1,\ldots, n\}$. We say that $f\in k$ has an {\em elementary integral over} $k$ if there exists an elementary extension $E$ of $k$ and $g\in E$ such that $Dg=f$. An {\em elementary function} is any element of any elementary extension of $(\mathbb{C}(x),d/dx)$.
\end{defn}

\begin{theorem}(Theorem 5.1.1 in \cite{BronsteinBook})
If $t$ is a primitive over a differential field $k$ and $Dt$ is not the derivative of an element of $k$, then $t$ is a monomial over $k$, $C_{k(t)}=C_k$, and $S=k$. Conversely, if $t$ is transcendental and primitive over $k$ and $C_{k(t))}=C_k$, then $Dt$ is not the derivative of an element of $k$.
\end{theorem}

\begin{theorem}(Theorem 5.3.1 in \cite{BronsteinBook})
Let $f\in k(t)$. Using only the extended Euclidean algorithm in $k[t]$, one can find $g,h,r\in k(t)$ such that $h$ is simple, $r$ is reduced, and $f=Dg+h+r$. Furthermore, the denominators of $g,h$ and $r$ divide the denominator of $f$, and either $g=0$ or $\mu(g)<\mu(f)$.
\end{theorem}

\begin{lemma}(Lemma 2.1 in \cite{CDL2018})
Let $g\in K[t]+K(t)'$. Then $g=0$ if it is $t$-simple.
\end{lemma}

\begin{theorem}(Theorem 3.1.1 (v) in \cite{BronsteinBook}, Logarithmic Derivative Identity)
Let $(R,D)$ be a differential ring. If $R$ is an integral domain, then $$\frac{D(u_1^{e_1}\cdots u_n^{e_n})}{u_1^{e_1}\cdots u_n^{e_n}}=e_1\frac{Du_1}{u_1}+\cdots+e_n\frac{Du_n}{u_n}$$ for any $u_1,\ldots,u_n\in R^*$ and any integers $e_1,\ldots,e_n$.
\end{theorem}

\begin{defn} (Definition 5.1.4 in \cite{BronsteinBook})
$K$ is an elementary extension of $k$ if there are $t_1,\ldots,t_n$ in $K$ such that $K=k(t_1,\ldots,t_n)$ and $t_i$ is elementary over $k(t_1,\ldots,t_{i-1})$ for $i\in[n]$. We say that $f\in k$ {\em has an elementary integral over} $k$ if there exists an elementary extension $E$ of $k$ and $g\in E$ such that $Dg=f$. An {\em elementary function} is any elementary extension of $(\mathbb{C}(x),d/dx)$.
\end{defn}

\begin{theorem}(Theorem 5.5.2 in \cite{BronsteinBook}, Liouville's Theorem)
Let $K$ be a differential field with an algebraically closed constant field and $f\in K$. If there exists an elementary extension $E$ of $K$ and $g\in E$ such that $Dg=f$, then there are $v\in K, u_1,\ldots,u_n\in K^*$ and $c_1,\ldots,c_n\in \mbox{Const}(K)$, such that $$f=Dv+\sum\limits_{i=1}^n c_i\frac{D u_i}{u_i}.$$
\end{theorem}

\begin{defn}(Chapter 2, Definition 3 in \cite{CLO}, Lexicographic Order)
Let $\alpha=(\alpha_1,\ldots,\alpha_n)$ and $\beta=(\beta_1,\ldots,\beta_n)$ be in $\mathbb{Z}_{\geq 0}^n$. We say $\alpha >_{lex}\beta$ if the leftmost nonzero entry of the vector difference $\alpha-\beta\in\mathbb{Z}^n$ is positive. We will write $x^{\alpha}>_{lex}x^{\beta}$ if $\alpha >_{lex}\beta$.
\end{defn}

\begin{cor}(Corollary 1' in \cite[Page 124]{ZariskiSamuel})
Let $K$ be a field and let $F=K(S)$ be a purely transcendental extension of $K$; here $S$ denotes a set of generators of $F/K$ which are algebraically independent over $K$. Let $x\rightarrow u_x$ be a mapping of $S$ into a field $L$ containing $F$. If $D$ is any derivation of $K$ with values in $L$, then there exists one and only one derviation $D'$ of $F$ extending $D$, such that $D(x)=u_x$ for all $x$ in $S$.
\end{cor}

\end{document}